\documentclass[useAMS,usenatbib,referee]{biom}

\usepackage{amsmath}

\usepackage{times}
\usepackage{bm}

\usepackage{bm} 

\usepackage{natbib}
\usepackage[margin=1cm]{caption}
\usepackage{color}
\usepackage{xcolor}

\usepackage[algo2e]{algorithm2e} 
\usepackage{defs}
\makeatletter
\renewcommand{\algocf@captiontext}[2]{#1\algocf@typo. \AlCapFnt{}#2} 
\def\@algocf@capt@plain{top}
\renewcommand{\algocf@makecaption}[2]{%
  \addtolength{\hsize}{\algomargin}%
  \sbox\@tempboxa{\algocf@captiontext{#1}{#2}}%
  \ifdim\wd\@tempboxa >\hsize
    \hskip .5\algomargin%
    \parbox[t]{\hsize}{\algocf@captiontext{#1}{#2}}
  \else%
    \global\@minipagefalse%
    \hbox to\hsize{\box\@tempboxa}
  \fi%
  \addtolength{\hsize}{-\algomargin}%
}
\makeatother

\title{Scalable Gaussian Process Regression via Median Posterior Inference for Estimating the Health Effects of an Environmental Mixture}

\author{Aaron Sonabend-W\\
	   Google Research, 1600 Amphitheatre Parkway, Mountain View, U.S.A.
    \and
	  	   Jiangshan Zhang\\
	   Department of Statistics, University of California, Davis,  One Shields Avenue, Davis, U.S.A.
            \and
            Edgar Castro\\ 
            Department of Environmental Health, Harvard University, 677 Huntington Ave, Boston, U.S.A.
            \and
        Joel Schwartz\\ 
	   Department of Environmental Health, Harvard University, 677 Huntington Ave, Boston, U.S.A.
    \and
     Brent A. Coull\\
	   Department of Biostatistics, Harvard University, 677 Huntington Ave, Boston, U.S.A.
    \and
    Junwei Lu\emailx{junweilu@hsph.harvard.edu}\\
	   Department of Biostatistics, Harvard University, 677 Huntington Ave, Boston, U.S.A.
	   }

\begin{document}




\begin{abstract}
Humans are exposed to complex mixtures of environmental pollutants rather than single chemicals, necessitating methods to quantify the health effects of such mixtures. Research on environmental mixtures provides insights into realistic exposure scenarios, informing regulatory policies that better protect public health. However, statistical challenges, including complex correlations among pollutants and nonlinear multivariate exposure-response relationships, complicate such analyses. A popular Bayesian semi-parametric Gaussian process regression framework \citep{Coull2015} addresses these challenges by modeling exposure-response functions with Gaussian processes and performing feature selection to manage high-dimensional exposures while accounting for confounders. Originally designed for small to moderate-sized cohort studies, this framework does not scale well to massive datasets. To address this, we propose a divide-and-conquer strategy, partitioning data, computing posterior distributions in parallel, and combining results using the generalized median. While we focus on Gaussian process models for environmental mixtures, the proposed distributed computing strategy is broadly applicable to other Bayesian models with computationally prohibitive full-sample Markov Chain Monte Carlo fitting. We apply this method to estimate associations between a mixture of ambient air pollutants and ~650,000 birthweights recorded in Massachusetts during 2001–2012. Our results reveal negative associations between birthweight and traffic pollution markers, including elemental and organic carbon and PM$_{2.5}$, and positive associations with ozone and vegetation greenness.
\end{abstract}

\maketitle

\section{Introduction}

Ambient air pollution consists of a heterogeneous mixture of multiple chemical components, with these components being generated by different sources.  Therefore, quantification of the health effects of this mixture can yield important evidence on the source-specific health effects of air pollution, which has the potential to provide evidence to support targeted regulations for ambient pollution levels.  

As is now well-documented, there are several statistical challenges involved in estimating the health effects of an environmental mixture. First, the relationship between health outcomes and multiple pollutants can be complex, potentially involving non-linear and non-additive effects. Second, pollutant levels can be highly correlated, but only some may impact health. Therefore,  models inducing sparsity are often advantageous.  Feature engineering, such as basis expansions to allow interaction terms, can lead to high dimensional inference. Alternatively, parametric models can be used, however they require the analyst to impose a functional form, which can yield biased estimates in the likely case that the model is miss-specified.


Several methods address the issues discussed above \citep{Billionnet2012}. A common approach to modelling the complex relationship between pollutants and outcomes is to use flexible models such as random forests which have been shown to be consistent \citep{ScornetErwan2015Corf}, or universal approximators, such as neural networks \citep{NN2015}. These are useful but yield results which are hard to interpret: one cannot report the directionality or magnitude of the feature effect on the outcome. In this context, our interest lies in both prediction as well as interpretation. Another possible way to incorporate flexible multi-pollutant modelling is by clustering pollution-exposure levels and including clusters as covariates in parametric models. This approach essentially stratifies exposure levels which results in important loss of information. It ultimately forces the analyst to adapt the question of interest into one that can be solved by available tools, instead of tackling the relevant questions. A common approach to address the high-dimensionality of multi-pollutants effects is to posit a generalized additive model. This allows one to estimate the association between a health outcome and a single pollutant, which can be repeated for every exposure of interest \citep{Stieb2012}. Flexible modelling such as quantile regression can be employed to deal with outliers and account for possible differences in associations across the health outcome \citep{Fong2019}. However, the clear downside is that incorporating multi-pollutant mixtures quickly makes this approach computationally infeasible. Alternatively, other parametric models \citep{GaskinsAudreyJ2019SFat,ijerph19031378,YU2022119356} have be used to evaluate the associations of interest, with the downside of imposing a functional form. To enforce sparsity on the feature space, variable selection methods such as least absolute shrinkage and selection operator (LASSO) penalty can be used \citep{Tibshirani1996}, however to use such methods one must specify a parametric model which brings us back to the likely misspecification scenario, in which estimated associations and causal effects may be biased.

A popular approach to simultaneously addressing these issues on small-scale data is the use of a semi-parametric Gaussian process model, often referred to as Bayesian kernel machine regression (BKMR)\citep{Bobb2015, Coull2015}. The pollutants-health outcome relationship is modelled through a Gaussian process, which allows for a flexible functional relationship between the pollutants and the outcome of interest. The model allows for feature selection among the pollutants to discard those with no estimable health effect and to account for high correlation among those with and without a health effect. This framework allows the incorporation of linear effects of baseline covariates, yielding an interpretable model.

Even though this framework is frequently employed in the multi-pollutant context, large datasets make it prohibitively slow as it involves Bayesian posterior calculation. To address this scalability challenge, we propose a divide-and-conquer approach in which we split samples, compute the posterior distribution, and then combine the smaller samples using the generalized median. This method allows capturing small effects from large datasets in little time. Our distributed algorithm is based on aggregating the median of the posteriors computed in parallel on the distributed datasets. Such a strategy is not only applicable to Gaussian process regression but also can be applied to a wider range of Bayesian methods. We provide theoretical guarantees for the convergence of the core Gaussian process component of the model, flexible to different function spaces. We then apply this scalable method to a challenging, large-scale dataset of 650,000 birthweights in Massachusetts to quantify the health effects of a mixture of ambient air pollutants.
\section{Method}

\subsection{Semi-parametric Regression}
Suppose we observe a sample of $n$ independent, identically distributed (i.i.d.) random vectors $\mathbb{S}_n=\{D_i\}_{i=1}^n$, where $D_i=(Y_i,\bX_i,\bZ_i )\sim P_0$ with $\bX_i\in\mathcal{X}\subset\mathbb{R}^p$ a vector of possible confounders, and $\bZ_i\in\mathcal{Z}\subset\mathbb{R}^q$ a vector of environmental exposure levels. We will assume health outcome $Y$ has a linear relationship with confounders $\bX$ and a non-parametric relationship with exposure vector $\bZ$. In particular, for $D_i$ we assume the following semi-parametric relationship:
\begin{align}\label{eq: model}
Y_i=\bX_i^\top\bbeta^0+h_0(\bZ_i)+e_i, 
\end{align}
where $e_i\sim\mathcal{N}(0,\sigma^2)$, and $h_0:\mathcal{Z}\mapsto\mathbb{R}$ is an unknown function which we allow to incorporate non-linearity and interaction among the pollutants. We require $h_0$ to be in an $\alpha$-H\:{o}lder space or to be infinitely differentiable. We formalize this in Section \ref{Section: theory}.

\subsection{Prior Specification}
To perform inference on $h_0$, we will use a re-scaled Gaussian process prior \citep{GPfM}. In particular, we will use a squared exponential process equipped with an inverse Gamma bandwidth. That is, we will use prior 
\[
h_0(\bZ)\sim\mathcal{N}(0,\bK),
\]
where Cov$[\bZ,\bZ']=K(\bZ,\bZ';\rho)=\exp\left\{-\frac{1}{\rho^{2}}\|\bZ-\bZ'\|_2^2\right\}$, and $\frac{1}{\rho^q}$ is a Gamma distributed random variable. We choose this kernel as it is flexible enough to approximate smooth functions, more so when the bandwidth parameter $\rho$ can be estimated from the data.

Alternatively, we can also augment the Gaussian kernel to allow for sparse solutions on the number of pollutants that contribute to the outcome \citep{Bobb2015}. Let the augmented co-variance function be Cov$[\bZ,\bZ']=K(\bZ,\bZ';\br)=\exp \{-\sum_{j=1}^{q}r_j(Z_j-Z_j')^2 \}$. To select pollutants we assume a "slab-and-spike" prior on the selection variables $r_j\sim g_{r\mid\eta}$, with
\[
g_{r\mid\eta}(r,\eta)=\eta f_1(r)+(1-\eta)f_0,\:\eta\sim\text{Bernoulli}(\pi),
\]
where $f_1$ has support on $\mathbb{R}^+$ and $f_0$ is the point mass at 0. The random variables $\eta_j \sim\text{Bernoulli}(\pi_j)$ can then be interpreted as indicators of whether exposure $Z_j$ is associated with the health outcome, 
and the variable importance for each exposure given the data is reflected by the posterior probability that $\eta_j = 1$.
Finally, for simplicity, we will assume an improper prior on the linear component: $\bbeta\sim1$. This linear component will capture the effects of confounders. We further use a Gamma prior distribution for the error term variance $\sigma^2$. 

\subsection{Estimation}\label{section: estimation}
Let $\bh = (h_0(\bZ_1),\dots,h_0(\bZ_n))^\top$, \citet{Liu2007} have shown that model \eqref{eq: model} can be expressed as 
\[
Y_i \sim \mathcal{N}(h_0(\bZ_i)+\bX_i^{T}\bbeta^0,\sigma^2),\:
\bh\sim\mathcal{N}(0,\tau \bK).
\]
This will allow us to simplify our inference procedure and split the problem into tractable posterior estimation \citep{Bobb2015} for each component of interest. In particular, we can use Gibbs steps to sample the conditionals for $\bbeta$, $\sigma^2$ and $\bh$ analytically. Letting $\lambda=\frac{\tau}{\sigma^2}$, we use a Metropolis-Hastings step, the full set of posteriors is given in equation \eqref{eq: posteriors}.
\begingroup
\allowdisplaybreaks
\begin{align}\label{eq: posteriors}
\bbeta\mid\sigma^2,\lambda,\br,\bY
&\sim N\left(\bV_{\bbeta} \bX^\top \bV^{-1}_{\lambda,\bZ,\br}\bY,\sigma^2\bV_{\bbeta}\right),
\notag\\
\sigma^2\mid\bbeta,\lambda,\br,\bY
&\sim \text{Gamma}\left(\alpha_{\sigma}+\frac{n}{2},b_{\sigma}+\frac{1}{2}WSS_{\bbeta,\lambda,\br}\right),\notag\\
{\bh}\mid\bbeta,\sigma^2,\lambda,{\br, \bY,\bX,\bZ}
&\sim N\left(\lambda \bK_{\bZ,\br}\bV^{-1}_{\lambda,{ \bZ,\br}}(\bY-\bX\bbeta),\sigma^2\lambda \bK_{\bZ,\br}\bV^{-1}_{\lambda,{ \bZ,\br}}\right),\\
 f\left(\mathbf{r}, \boldsymbol{\eta}, \mid \boldsymbol{\beta}, \sigma^2, \lambda, \mathbf{y}\right) & \propto \Gamma\left(\sum_j \eta_j+a_{\pi}\right) \Gamma\left(q-\sum_j \eta_j+b_{\pi}\right)\left\{\prod_{j=1}^q f\left(r_j \mid \eta_j\right)\right\},\notag\\
f(\lambda\mid\bbeta,{\br,\eta, \bY,\bX,\bZ})&\propto \left|\bV^{-1}_{\lambda,{\bf Z,r}}\right|^{-1/2}\exp\left\{-\frac{1}{2\sigma^2}WSS_{\bbeta,\lambda,\br}\right\}\text{Gamma}(\lambda\mid a_{\lambda},b_{\lambda})
,\notag
\end{align}
\endgroup
where $\bV_{\lambda,{\bZ,\br}}={\bf I}_n+\lambda{\bK}_{\bZ,\br}, {\bV_{\bbeta}}=(\bX^\top \bV^{-1}_{\lambda,{\bf Z,r}}\bX)^{-1}, WSS_{\bbeta,\lambda,{\bf r}}=(\bY-\bX\bbeta)^\top \bV^{-1}_{\lambda,{\bf Z,r}}(\bY-\bX\bbeta).$ 

To perform inference for functions of interest in \eqref{eq: posteriors}, we will use Markov Chain Monte Carlo (MCMC) techniques. Furthermore, even though function $\bh$ has a closed-form posterior, large samples will require large matrix inversions. 

Posterior sampling can be challenging for this model. This is particularly true for sampling $\bh$, as the posterior Gaussian process is $n$-dimensional. 
First, in practice the number of samples for the burn-in state required from the true posterior significantly increases with dimension.  Second, the problem worsens as the sample size grows.  The computational cost for each iteration is $\mathcal{O}(n^3)$, since to sample from the posterior of $\bh$ we need to compute an inverse of an $n\times n$ kernel matrix indicated in \eqref{eq: posteriors}. This renders the method prohibitively slow for real applications on large data sets. On the other hand  these are precisely the data sets needed as they can actually shed light on the small effects of environmental mixtures on health outcomes. This predicament motivates the development of a computationally fast version of inference for \eqref{eq: posteriors}, and particularly for $\bh$.

\subsection{Fast Inference on Posteriors via Sub-sampling}

In order to make posterior sampling computationally feasible, we propose a sample splitting technique which is guaranteed to satisfy the needed theoretical properties. Our approach consists of computing multiple \textit{noisy} versions of the posteriors we are interested in, and using the median of these as a proxy for the full data posterior. 

First, we randomly split the entire data set into several disjoint subsets with roughly equal sample size without replacement. Let $\mathbb{S}_{k=1}^K$ denote a random partition of $\mathbb{S}_n$ into $K$ disjoint subsets of size $n_k={n}/{K}$ with index sets $\{\mathcal{I}_k\}_{k=1}^K$. Then for each subset $\mathbb{S}_{k}$, we run a modified version of the estimation approach described in Section \ref{section: estimation} using sub-sampling sketching matrices $S_k \in \mathbb{R}^{n\times n_k}$. This will yield $K$ posterior distributions for each parameter and function in \eqref{eq: posteriors}.

We define the $n \times n_k$ sketching matrix $S_k$ with its $i$th column $S_{k,i}=\sqrt{K} \cdot  p_i$, where $p_i$ is uniformly sampled from the columns of the identity matrix. Using $S_k$ we denote by $\tilde\bV_k$ and $\tilde\bA_k$, any vector and matrix transformation respectively as $\tilde\bV_k=S_k^\top\bV$, $\tilde\bA_k=S_k^\top\bA S_k$. In specific, $\tilde \bh_k =  S_k^\top\bh, \tilde \bY_k =  S_k^\top\bY, \tilde \bX_k =  S_k^\top\bX$, and $\tilde \bK_k =  S_k^\top\bK S_k$. We can then redefine model \eqref{eq: model} for sample $\mathbb{S}_k$ as
\begin{align}\label{modif}
\begin{split}
&\tilde{\bY}_k\sim N(\tilde{\bh}_k+\tilde{\bX}_k\bbeta,\sigma^2S_k^\top S_k),\;\; \tilde{\bh}_k\sim N({\bf 0},\tau \tilde{\bK}_k).
\end{split}
\end{align}
We then implement our inference from Section \ref{section: estimation} to the above by using
 $\tilde \bV_{\bbeta}^{(k)}=(\tilde \bX_k^\top(\tilde \bV^{(k)}_{\lambda,\bZ,\br})^{-1}\tilde \bX_k)^{-1},$ $\tilde \bV_{\lambda,\bZ,\br}^{(k)} = S_k^\top S_k + \lambda \tilde{\bK}_k,
$
$\tilde{WSS}_{\bbeta,\lambda,\br}^{(k)}=(\tilde \bY_k-\tilde \bX_k\beta)^\top(\tilde \bV^{(k)}_{\lambda, \bZ,\br})^{-1}(\tilde \bY_k-\tilde \bX_k\bbeta)
$
in \eqref{eq: posteriors} and sample from each of the $K$ posteriors.

Let $\Theta$ be a parameter space and let $\{P_\theta|\theta \in \Theta\}$ be a set of probability measures on $\mathbb{R}^q$ which is indexed by $\Theta$, and which are absolutely continuous with respect to the Lebesgue measure $\mu$. For any i.i.d. random vectors $D_1,\dots,D_n\sim P_0$, let $P_0\equiv P_{\theta_0}$. Bayesian inference usually consists of specifying a prior distribution $\Pi$ for $\theta$, and using sample $\mathbb{S}_n$ to compute a posterior distribution for $\theta$ defined as 
\[
\Pi_n(\theta\mid\mathbb{S}_n)\equiv\frac{\prod_{i=1}^np_\theta(D_i)\Pi(\theta)}{\int_\Theta\prod_{i=1}^np_\theta(D_i)\Pi(\theta)},
\]

\noindent where $p_\theta d\mu = dP_\theta$.

Note that this definition is general enough that $\theta$ can be any parameter in \eqref{eq: posteriors} as well as function $h_0$, in which case prior $\Pi(\theta)$ is a Gaussian process. 
Thus, we compute $\Pi_k(\theta\mid\mathcal{I}_k)$  for each split $k=1,\dots,K$ and each parameter of interest. 
Naturally, posterior $\Pi_k(\theta\mid\mathcal{I}_k)$ will be a noisy approximation of $\Pi_n(\theta\mid\mathbb{S}_n)$. We denote $\Pi_k = \Pi_k(\theta\mid\mathcal{I}_k)$ for notation simplicity. We combine each $\Pi_k$ using the geometric median. This aggregation framework is general and not restricted to parametric models; it can be applied to probability measures on abstract spaces, including non-parametric function spaces. To formalize this, we first define the geometric median, which is a multi-dimension generalization of the univariate median \citep{Minsker2015}.

To construct the geometric median posterior, first define $\psi$ to be the Hellinger metric on $\Theta$ such that, for $\theta_1, \theta_2 \in \Theta$

\[
\psi(\theta_1, \theta_2) = \sqrt{\frac{1}{2}\int_{\mathbb{R}^q}\big(\sqrt{p_{\theta_1}(x)} - \sqrt{p_{\theta_2}(x)}\big)^2 d\mu(x)} 
\]

\noindent and let $(\Theta, \psi)$ be a separable metric space. Now, let $(\mathbb{H}, \langle\cdot, \cdot\rangle_\mathbb{H})$ be the Reproducing Kernel Hilbert Space of functions $f:\Theta \rightarrow \mathbb{R}$ with characteristic reproducing kernel $\kappa: \Theta \times \Theta \rightarrow \mathbb{R}$ defined as the $L^2$ inner product of $p_{\theta_1}$ and $p_{\theta_2}$ with respect to $\mu$. Then, let $\mathbb{H}_1$ be the defined as the unit ball in $\mathbb{H},$ that is, $\mathbb{H}_1 = \{f\in \mathbb{H}| \ ||f||_\mathbb{H} \leq 1\}$. 

Now, let $\mathcal{P}$ be the space of probability measures over $\Theta$ and denote the space

\[
\mathcal{P}_\kappa = \bigg\{\Pi \in \mathcal{P}\bigg|\int_\mathcal{P}\sqrt{\kappa(x, x)} d\Pi(x) < \infty\bigg\}.
\]

\noindent We define the distance between probability measures $\Pi_1, \Pi_2 \in \mathcal{P}_\kappa$ by 

\[
||\Pi_1 - \Pi_2||_{\mathbb{H}_1} = \bigg|\bigg|\int_\Theta \kappa(x, \cdot) d(\Pi_1 - \Pi_2)(x)\bigg|\bigg|_\mathbb{H}.
\]

\noindent With this notion of distance for priors and posteriors on $\Theta$, we now define the geometric median by 

\begin{equation}\label{barycenter}
    \bar\Pi_n = \argmin_{\Pi \in \mathcal{P}_\kappa} \sum_{k = 1}^K||\Pi - \Pi_k||_{\mathbb{H}_1}.
\end{equation}

It is important to note that this definition operates on a general parameter space $\Theta$, which in our case is a non-parametric function space. Our approach applies this general, non-parametric aggregation theory directly to the posterior measures for the function $h_0$.

As shown in \citet{Minsker2017}, the geometric median $\overline{\Pi}_n$ is robust to outlier observations, meaning that our estimation procedure will be robust to outliers.

The convergence rate will be in terms of $n_k$, however this rate improves geometrically with $K$ with respect to the rate at which the $K$ estimators are weakly concentrated around the true parameter \citep{Minsker2017}.

As the median function $\overline{\Pi}_n$ is generally analytically intractable, an achievable solution is to estimate $\overline{\Pi}_n$ from samples of the subset posteriors. We can approximate the median function by assuming that subset posterior distributions are empirical measures and their atoms can be simulated from the subset posteriors by a sampler \citep{SrivastavaSanvesh2015WSBv}. Let $\{\btheta_{k1},\ldots,\btheta_{kN}\}$ be N samples of parameters $\btheta$ obtained from subset posterior distribution $\Pi_n^{(k)}$. In our method, samples can be directly generated from subsets posteriors by using an MCMC sampler. Then we can approximate the $\Pi_n^{(k)}$ by the empirical measure corresponding with $\{\btheta_{k1},\ldots,\btheta_{kN}\}$, which is defined as:
\begin{align}\label{atom subset}
    \hat{\Pi}_n^{(k)}(\cdotp)=\sum_{i=1}^N \frac{1}{N}\delta_{\btheta_{ki}}(\cdotp),\:(k=1,\dots,K),
\end{align}
where $\delta_{\btheta_{ki}}(\cdotp)$ is the Dirac measure concentrated at $\btheta_{ki}$. In order to approximate the subset posterior accurately, we need to make $N$ large enough. Then the empirical probability measure of median function $\hat\Pi_{n,\mathcal{I}}$ can be approximated by estimating the geometric median of the empirical probability measure of subset posteriors. Using samples from subsets posteriors, the empirical probability measure of the median function is defined as:
\begin{align}\label{atom full}
    \hat{\overline{\Pi}}_{n,\mathcal{I}}(\cdotp)=\sum_{k=1}^K\sum_{i=1}^N a_{ki}\delta_{\btheta_{ki}}(\cdotp),\:0\leq a_{ki}\leq 1,\:\sum_{k=1}^K\sum_{i=1}^Na_{ki}=1
\end{align}
where $a_{ki}$ are unknown weights of atoms.  Here the problem of combining subset posteriors to give a problem measure is switched to estimating $a_{ki}$ in \eqref{atom full} for all the atoms across all subset posteriors. Fortunately, $a_{ki}$ can be estimated by solving the optimization problem in \eqref{barycenter} via kernel trick with posterior distributions restricted to atom forms in \eqref{atom subset} and \eqref{atom full}. There are several different algorithms to solve this, such as \citet{Bose2003-cz} and \citet{Cardot2013-qw}. We use an efficient algorithm developed by \citet{Minsker2017}, which is summarized as its Algorithm 2 via Weiszfeld's algorithm.

\begin{algorithm}
\caption{Fast Posterior Inference Via Sub-sampling.}
\label{alg1}
\begin{tabbing}
   \qquad \enspace REQUIRE Observed sample $\mathbb{S}_n=\{D_i\}_{i=1}^n$, subset number K, parameter sample size N\\
   \qquad \qquad  Randomly partition $\mathbb{S}_n$ without replacement into $K$ subsets $\mathbb{S}_{k=1}^K$ with size $n_k$ \\
   \qquad \qquad  For $\mathbb{S}_k \in \mathbb{S}_{k=1}^K$ do\\
   \qquad \qquad  \qquad 1. Get index set $\mathcal{I}_k$ for $\mathbb{S}_k$ \\
   \qquad \qquad  \qquad 2. Get sub-sampling sketching matrix $S_k$ \\
   \qquad \qquad  \qquad 3. Run MCMC sampling on modified model described as \eqref{eq: posteriors} and \eqref{modif} to\\
   \qquad \qquad  \qquad generating parameter samples $\{\btheta_{k1},\ldots,\btheta_{kN}\}$\\
   \qquad \qquad  Solve the linear program in \eqref{barycenter} using Weizfeld's Algorithm with \eqref{atom subset}-\eqref{atom full}\\
   \qquad \qquad RETURN empirical approximation posterior median function $\hat{\overline{\Pi}}_{n,\mathcal{I}}$
\end{tabbing}
\end{algorithm}

Algorithm \ref{alg1} provides a sample splitting approach to decrease computational complexity for posterior inference. Sampling $\bbeta$, $\sigma^2$, $\lambda$ and $\bh$ requires computing $K$, and inverting $V_{\lambda,{\bf Z,r}}$, this translates into $O(n^2q)$, $O(n^3)$ operations respectively per iteration. For $\lambda$ we also need to compute $\left|V_{\lambda,{\bf Z,r}}\right|$ which is $O(n^3)$. \citet{Bobb2015} recommend using at least $10^4$ iterations, which translates into $O(10^4n^3)$ operations (assuming $q<<n$). There is a clear trade-off between a large number of splits $K$ which decreases computational complexity, and using the whole sample $K=1$ which yields better inference. For example, choosing $K=n^{1/2}$ yields a computational complexity of  $O(10^4n^{3/2})$ for Algorithm \ref{alg1}. Next in Section \ref{Section: theory} we discuss the posterior convergence rate on a simpler, special case, and its dependence on the number of splits in detail.
\section{Theoretical Results}\label{Section: theory}

In this section we present the assumptions needed for our theoretical results, state our main theorem and discuss its implications. Our results focus on estimation of $h_0$ as this is the main function of interest and our main contribution and do not extend to the variable selection parameters $r_1, \ldots, r_q$. A complete theoretical analysis of the convergence for the distributed ``slab-and-spike" model remains a non-trivial challenge for future work. In the rest of this section, we will consider the revised model by removing the variable selection parameters $r_1, \ldots, r_q$ in \eqref{eq: posteriors} and focus on the estimation of $h_0$ only.

The proof structure logically layers two distinct nonparametric theories: first, we use established results on Gaussian Process posterior contraction rates for a single subset (e.g., \cite{van_der_Vaart_2009}); second, we use these rates as inputs for the general, nonparametric theory of median posterior convergence \citep{Minsker2017} to derive the final rate for our combined estimator. By proving the convergence of the nonparametric component $h_0$, our theorem provides the necessary theoretical guarantee for the semi-parametric model.

We first assume that the confounder and pollution exposure space $\mathcal{X}$, $\mathcal{Z}$ respectively are compact bounded sets. This is easily satisfied in practice. Next we define two function spaces. Letting $\alpha>0$, we define $C^\alpha[0,1]^q$ to be the Holder space of smooth functions $h:[0,1]^q\mapsto\mathbb{R}$ which have uniformly bounded derivatives up to $\floor\alpha$, and the highest partial derivatives are Lipschitz order $\alpha-\floor\alpha$. More precisely, we define the vector of $q$ integers as $\bv=(v_1,\dots,v_q)$ and 
\[
D^vg(\bz)=\frac{\partial^{\left(\sum_iv_i\right)}g(\bz)}{\partial z_1^{v_1},\dots,\partial z_q^{v_q}}.
\]
Then for function $h$ we define
\[
\|h\|_\alpha=\max_{\sum_iv_i\le\floor\alpha}
\sup_{D^vg(\bz)}+\max_{\sum_iv_i=\floor\alpha}\sup\frac{\left|D^vg(\bz)-D^vg(\bz')\right|}{\|\bz-\bz'\|^{\alpha-\floor\alpha}},\text{ for} \bz\neq\bz'.
\]
With the above we say that 
\[
C^\alpha[0,1]^q=\left\{h:[0,1]^q\mapsto\mathbb{R}\bigg|\|h\|_\alpha<M\right\}
\]
\citep{VaartAadW.1996WCaE}. 
 
Note that $C^\alpha[0,1]^q$ might be too large of a space as it is highly flexible in terms of differentiability restrictions. In light of this, if we only consider smooth functions, we introduce the following space. 

Let the Fourier transform be $\hat h(\lambda)=1/(2\pi)^q\int e^{i(\lambda,t)}h(t)dt$ and define \[
\mathcal{A}^{\gamma,r}(\mathbb{R}^q)=\left\{h:\mathbb{R}^q\mapsto\mathbb{R}\bigg|\int e^{\gamma\|\lambda\|^r}|\hat h(\lambda)|^2d\lambda<\infty\right\}.
\]
Set $\mathcal{A}^{\gamma,r}(\mathbb{R}^q)$ contains infinitely differentiable functions which are increasingly smooth as $\gamma$ or $r$ increase \citep{van_der_Vaart_2009}. 

\begin{theorem} \label{theorem: median rate}
Let $\delta_0(h)$ be the Dirac measure concentrated at $h_0$. For any $\delta  \in (0,1)$,  there exists a constant $C_1$ such that if we choose the number of splits  $K\le C_1\log{1/{\delta}}$, then with a probability of at least $1- \delta$, we have
\[
\|\hat\Pi_{n,g} - \delta_0(h_0)\|_\mathcal{F} \le 
\left\{
\begin{array}{ll}
    C_2(n/\delta)^{-\left(\frac{\alpha}{2\alpha+q}\right)}\left(\log (n/\delta)\right)^{\left(\frac{4\alpha+q}{4\alpha+2q}\right)} & \mbox{ if }  h_0\in\mathcal{C}^\alpha[0,1]^q,\\
   C_3 (n/\delta)^{-\frac{1}{2}}\left(\log (n/\delta)\right)^{\left(q+1+q/(2r)\right)} & \text{ if } h_0\in\mathcal{A}^{\gamma,r}(\mathbb{R}^q) \text{ and }r<2,\\
   C_3 (n/\delta)^{-\frac{1}{2}}\left(\log (n/\delta)\right)^{\left(q+1\right)} & \text{ if } h_0\in\mathcal{A}^{\gamma,r}(\mathbb{R}^q) \text{ and }r\ge2,
\end{array}
\right.
\]
where $C_2,C_3$ are sufficiently large constants.
\end{theorem}

The proof follows from results on convergence of the posterior median and scaled squared exponential Gaussian process properties. We defer the proof to the supplementary materials. 
The rate in Theorem \ref{theorem: median rate} is achieved for all levels of regularity $\alpha$ simultaneously. If $h_0\in\mathcal{C}^\alpha[0,1]^q$, then the adaptive rate is $\tilde{\mathcal{O}}\left((n/\delta)^{-\left({\alpha}/({2\alpha+q}\right)})\right)$, however further assuming $h_0$ is infinitely differentiable, then $h_0\in\mathcal{A}^{\gamma,r}(\mathbb{R}^q)$ and we recover the usual $\tilde{\mathcal{O}}\left(n^{-1/2}\right)$ rate. Intuitively, understanding $\alpha$ as the number of derivatives of $h_0$, this $n^{-1/2}$ rate is obtained letting $\alpha\rightarrow\infty$. 
Theorem \ref{theorem: median rate} sheds light into the trade-off between choosing the optimal number of splits $K$: large $K$ negatively impacts the statistical rate as it slows down convergence, however it helps with respect to computation complexity. Finally, dimension $q$ affects the rate on a logarithmic scale if $h_0$ is infinitely differentiable; in the case that $h_0\in\mathcal{C}^\alpha[0,1]^q$ then $q$ has a larger effect in the rate. This trade-off is further illustrated in Section \ref{section: simulation}.
 
\section{Simulation Results} \label{section: simulation}
To study our method's empirical performance in finite samples we evaluated it in several simulation settings. The simulated data is generated with the following procedure. We generated data sets with $n$ observations, $\mathbb{S}_n=\{D_i\}_{i=1}^n$, $D_i=(y_i,x_i,\bz_i)$, where $\bz_i=(z_{i1},\dots,z_{iq})^\top$ is the profile for observation $i$ with $q$ mixture components and $x_i$ is a confounder of the mixture profile generated by $x_i\sim N(3 \cos (z_{i1}),2)$. The outcomes were generated by $y_i \sim N(\beta^0 x_i + h_0(\bz_i),\sigma^2)$. Given the main focus is on how the algorithm performs for large $n$, we considered $q=4$ exposures, and each exposure vector $\{\bz_i\}_{i=1}^n$ was obtained by sampling each component $\bz_{i1},\dots,\bz_{i4}$ from the standard normal distribution $N(0,1)$. 

We considered the mixture-response function $h_0(\cdotp)$ as a non-linear and non-additive function of only ($z_{i1},z_{i2})$ with an interaction. In particular, let $\phi(x)=1/(1+e^{-x})$.  We generated $h$ as $$
h_0(\bm{z}_i)=4\phi\left(\frac{5}{6}\left\{z_{i1}+z_{i2}+\frac{1}{2}z_{i1}z_{i2}\right\}\right).
$$
We set $\beta^0=2$, and $\sigma^2=0.5$. We considered the total number of samples $n\in\{512, 1024, 2048, 4096\}$, and the number of splits $K=n^t$, with $t\in\{0,0.05,0.1,0.15,\ldots,0.7\}$. Note that for $n=512$, the subset sample size is not enough for performing the MCMC sampler when $t=0.7$. Therefore, simulation under this particular combination of $n=512$, $t=0.7$ was not performed. 
Each simulation setting is replicated 300 times. All computations are performed on a server with 2.6GHz 10 core compute nodes, with 15000MB memory.

\subsection{Assessing Performance Across Split Levels }\label{simu split levels}
Figures \ref{fig:reg_h} and \ref{fig:time} show the method performance for approximating $h_0$ by the median posterior. To evaluate performance, we ran a linear regression of the  posterior mean $\hat h$ on true $h_0$, i.e. $h_0(\bz_i)=\gamma_0+\gamma_1\hat h(\bz_i)+\epsilon_i$ $i=1,\dots,n$ and plot the estimated slope $\hat\gamma_1$, intercept $\hat\gamma_0$ and $R^2$ while varying number of sample splits $K$. A good $\hat h(\cdotp)$ would yield $\hat\gamma_0=0$, $\hat\gamma_1=1$, $R^2=1$ as $h_0(\bz)\approx\hat h(\bz)$. As the figure shows, as the number of splits increases with $t$, inference on $h_0$ starts to lose precision. This is natural; although the median geometrically improves at rate $n_k$, as splits increase each posterior sample becomes noisier. However, near $t=1/2$ the median performance for $\hat h$ is close to the performance when the entire sample is used ($t=0$) as measured by $\hat\gamma_0,\hat\gamma_1,R^2$, with significant computation time gains. Figure \ref{fig:time} shows computing time for inference on $h_0$ through the posterior median. There is a trade-off between sampling from a high dimensional Gaussian process posterior of $n$ samples, and a large number of data splits which require almost equivalent computation power to sample. However, this can be mitigated when datasets are sampled or collected distributedly.   Results suggest that splits with $t\in[1/4,1/2]$ decrease computational burden significantly. On the other hand Figure \ref{fig:reg_h},\ref{fig:time} and theoretical results in Section \ref{Section: theory} suggest that $t\le1/2$ offers a good approximations to the full-data posterior. Figures 1 and 2 and the theoretical results suggest choosing $t\in[1/4,1/2]$, with $t\rightarrow1/2$ as $n$ increases will optimize the computation-cost vs. precision trade-off. In the meantime, the mean squared error (MSE) of the estimated response $h$ was reported as Figure A4 
in the supplementary materials, which provides the same result as Figures \ref{fig:reg_h} and \ref{fig:time}.
\begin{figure}[H]
  \centering
  \includegraphics[scale=0.32]{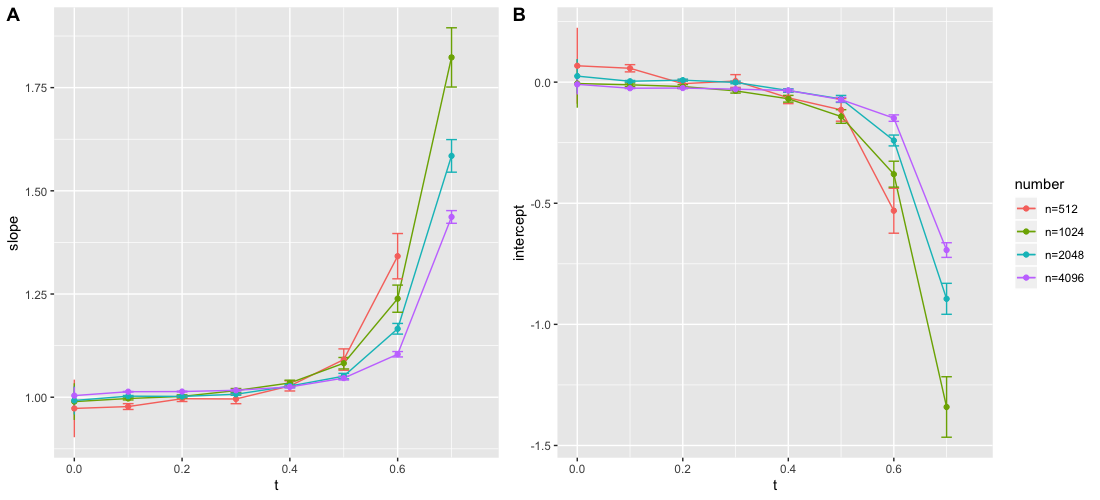}
   \caption{Regression summary results for ${\bf h} = \gamma_0 + \gamma_1 {\bf\hat h}$ across different sample size $n$ and data set splits. The setting of number of subsets are described above as $n^t$. We show (A) intercept: $\hat\gamma_0$, (B) slope: $\hat\gamma_1$. } 
   \label{fig:reg_h}
\end{figure}

\begin{figure}[H]
   \centering
   \includegraphics[scale=.32]{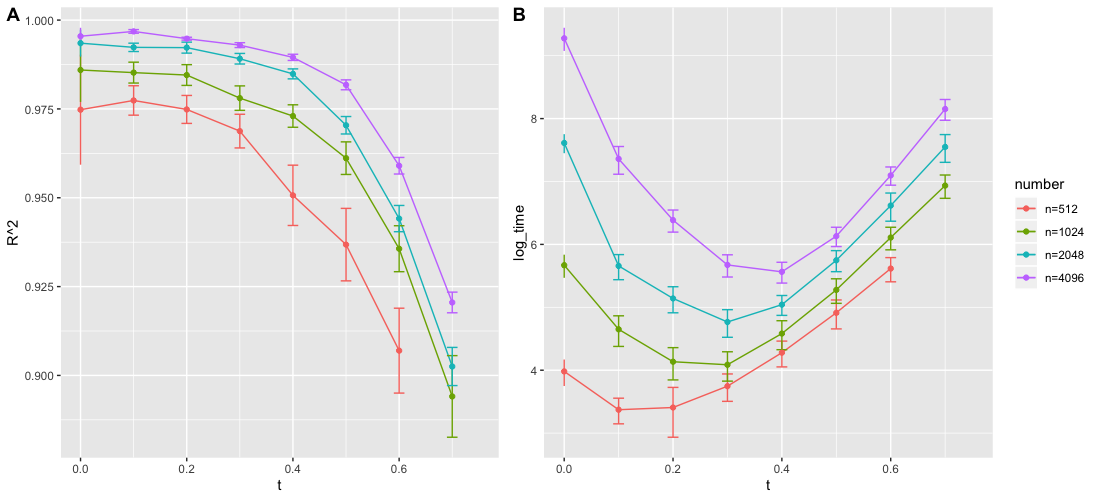}
   \caption{(A)Regression $R^2$ for ${\bf h} = \gamma_0 + \gamma_1 {\bf\hat h}$ and (B) Logarithmic runtime for the proposed method across different sample size $n$ and data set splits. The setting of number of subsets are described above as $n^t$. }
   \label{fig:time}
\end{figure}

To assess credible intervals for $\hat{h}$ estimation,  we reported mean coverage probability (MCP) of 95\% credible intervals of $\hat{h}$, which is calculated as $\text{MCP} = \frac{1}{U}\frac{1}{n}\sum_{u=1}^{U}\sum_{i=1}^{n} I(\hat{h_i}^{(u)})$, where $I(\cdot)$ is the indicator function for whether $h_i$ falls within the estimated credible interval. As shown in Table \ref{tab:CP}, for all the sample size settings, as the number of splits increases with t near 1/2, the MCP remains close to the nominal 95\%. Once $t>1/2$, coverage deteriorates sharply. The result is consistent with the performance shown in Figure \ref{fig:reg_h}. Moreover, as n increases the estimated standard deviations shrink, which amplifies the under-coverage for $t>1/2$; in our largest-n settings, MCP can approach zero under this regime.
\begin{table}
    \centering
    \begin{tabular}{ccccccccc}
        n & \multicolumn{8}{c}{t}  \\
        \hline
         ~& 0 &0.1 &0.2 &0.3 &0.4 &0.5 &0.6 &0.7\\
         \hline
         512 & 0.9628 & 0.9863 & 0.9808 & 0.9648 & 0.6015 & 0.3937 & 0.0937 & NA\\
         1024 & 0.9267 & 0.9218 & 0.9541 & 0.9414 & 0.8769 & 0.6943 & 0.0802  & 0.0003\\
         2048 & 0.9877 & 0.9431 & 0.9609 & 0.9897 & 0.8442 & 0.8334 & 0.0522 & 0.0004\\
         4096 & 0.9652 &0.9522 & 0.9248 & 0.9609 & 0.9306 & 0.9324 & 0.0576 & 0.0000
    \end{tabular}
    \caption{Mean Coverage Probabilty of 95\% Credible Interval of $\hat{h}$}
    \label{tab:CP}
\end{table}

 Variable selection performance of our method was also evaluated with the mean posterior inclusion probabilities (PIPs). We provided the mean PIPs for $X_3$ and $X_4$ cross sample sizes and split exponents t as Table A2 
 in the supplementary materials. The mean PIPs for both variables remain near zero as $t \xrightarrow{}1/2$. When $t>1/2$, the precision of PIPs decreases. The same qualitative pattern holds for$X_1$ and $X_2$: PIPs are exactly 1 when $t$ is small, and they decline modestly to around 0.8 once $t>1/2$.
 
 The performance of the method under a more realistic exposure setting where exposures are highly correlated is reported in the supplementary material.

\subsection{Assessing Performance Across Heterogeneous Subsetting Scheme }\label{simu subsetting scheme}
Another concern about the algorithm consistency is whether the subsetting scheme will affect the performance. Specifically, how heterogeneous splitting will affect the estimation accuracy. As for the subset sample size variability, we assess the model performance via the following additional simulation. We considered the same simulation setting as section \ref{simu split levels}, with a total number of samples $n=2048$. The number of splits $K=n^t$ is fixed with $t \in \{0.1,...,0.4\}$. Instead of randomly evenly partitioning the sample into K subsets, we considered the sample size for each subset to be randomly sampled from $[(1-\Delta)*n/K, (1+\Delta)*n/K]$, with $\Delta \in \{0,0.1,...,0.7\}$. Each simulation setting is replicated 300 times. We reported the MSE of posterior mean $\hat{h}$ in Figure A5 in the supplementary materials. As the figure shows, when the number of splits is relatively small, inference on $h_0$ keeps the same level of precision as the variability added into the subset sample sizes. The performance is mainly affected by the smallest subset, and it can be seen that when $\Delta>0.7$ and $t=0.4$, estimation accuracy decreases significantly. We thus still recommend ensuring the smallest sample size of subsets should be at least $n^{0.5}$. A more extreme uneven partition scenario is investigated and reported in the supplementary materials.

As for the subset exposure covariate distributions variability, the model variable selection performance was assessed via the following additional simulation. We considered the same simulation setting as section \ref{simu split levels}, with a total number of samples $n=2048$. The number of splits $K=n^t$ is fixed with $t=0.2$. To include scale heterogeneity for covariates, for the first $K/2$ subsets, we randomly select one of four exposure variables and shrink its marginal variance by a factor $c_v$, with $c_v \in \{0.1^2,0.3^2,0.5^2,0.7^2\}$. For the remaining $K/2$ subsets, exposure variances are unchanged. Each simulation setting is replicated 300 times. The method performance was compared with the traditional BMKR fitted to the whole dataset($t=0$). We reported the mean overall PIPs and mean across-subset PIP standard deviations(SD) for exposure variables as Table \ref{tab:PIP_var_covarite}. Across all scenarios, PIPs were highly stable for $X_1$ and $X_2$. While the mean PIPs and mean across-subset PIPs SDs of two noise exposure variables $X_3$ and $X_4$ increase as the level of covariate heterogeneity increases, they still present reliable variable selection behavior with conventional thresholds (e.g., 0.5). Overall, variable selection is thus robust to moderate subsetting-induced heterogeneity in exposure distributions when a reasonable number of subsets is used.
\begin{table}[!htp]
    \centering
    \begin{adjustbox}{width=0.9\textwidth,center}
    \begin{tabular}{ccccccccccc}
        $c_v$ & \multicolumn{2}{c}{0.1} &\multicolumn{2}{c}{0.3} &\multicolumn{2}{c}{0.5}&\multicolumn{2}{c}{0.7}&\multicolumn{2}{c}{1} \\
        \hline
        \hline
         t& 0 &0.2 &0 &0.2 &0 &0.2 &0 &0.2&0 &0.2\\
         \hline
         $X_1$ & 1 & 1(0.000) & 1 & 1(0.000) & 1 & 1(0.000) & 1 & 1(0.000)& 1 & 1(0.000)\\
         $X_2$ & 1 & 1(0.000) & 1 & 1(0.000) & 1 & 1(0.000) & 1  &1(0.000)& 1 & 1(0.000)\\
         $X_3$  & 0.088 & 0.291(0.143) & 0.076 & 0.284(0.125) & 0.022 & 0.182(0.090) & 0.008& 0.159(0.076) &0.004 &  0.006(0.046)\\
         $X_4$ & 0.084 & 0.251(0.127) & 0.042 & 0.173(0.094) & 0.012 & 0.112(0.075)& 0.008 & 0.055(0.053) & 0.004 & 0.004(0.024)
    \end{tabular}
    \end{adjustbox}
    \caption{Mean estimated PIP and mean across-subset standard deviations of exposure variables across different shrinkage factors $c_v$ of exposure variance. }
    \label{tab:PIP_var_covarite}
\end{table}

\section{Application:  Major Particulate Matter Constituents and Greenspace on Birthweight in Massachusetts}

To further evaluate our method on a real data set, we considered data from a study of major particulate matter constituents and greenspace and birthweight. The data consisted of the outcome, exposure and confounder information from 907,766 newborns in Massachusetts born during January 2001 to 31 December 2012 (\cite{FongKelvinC2019Rtom}). After excluding the observation records with missing data, there were $n=685,857$ observations used for analysis. As exposures, we included exposure data averaged over the pregnancy period on address-specific normalized Ozone,  PM$_{2.5}$, EC, OC, nitrate, and sulfate, as well as the normalized difference vegatation index (NDVI), in the nonparametric part of the model, and confounders including maternal characteristics in the linear component of the model. These pollutant exposures were estimated using a high-resolution exposure model based on remote-sensing satellite data, land use and meterologic variables \citep{di2016hybrid}. We randomly split the sample to $K=686$ (using $t\approx1/2$) different splits, each split contains $\approx$1000 samples. For each split, we ran the MCMC sampler for 1,000 iterations after 1,000 burn-in iterations, and every fifth sample was kept for further inference, thus we retain $N=200$ posterior samples for each split. Further details on the confounders included in the analysis can be found in the supplementary materials.
In addition, the correlation matrix for this data is in Figure A1 in the supplementary materials.

Figure \ref{fig: pollutant exp}  shows estimated exposure-response relatoinships between the mean outcome and each exposure, when fixing all other exposures at their median values. Results suggests that, for the PM$_{2.5}$, EC, and OC terms, exposure levels are negatively associated with mean birthweight.  In contrast,  Ozone, nitrate, and NDVI levels are positively associated with mean birthweight, and there is no association between birthweight and maternal exposure to sulfate. Among negatively associated constituents, EC and remaining PM$_{2.5}$ constituents have stronger linear negative associations compared to OC. Among positive associations, NDVI and Ozone seem to have a strong linear relationship with birth-weight. However, for nitrate, when its concentration is lower than +1 standard deviation, it is positively associated with birth weight increase, whereas when it is above the mean level over 1 standard deviation, it is negatively associated with birth-weight. This suggests the possibility of effect modification.

Figure \ref{fig: pollutant bivariate} investigates the bivariate relationship between pairs of exposures and birthweight, with other exposures fixed at their median levels. Results suggest different levels of non-linear relationships between constituent concentrations and birthweight. Unlike the pattern of sulfate shown in figure \ref{fig: pollutant exp}, there exists a strong inverted u-shaped relationship between sulfate and mean birthweight when nitrate concentration is at around -1 standard deviation. A similar relationship is visible between nitrate and mean birthweight when sulfate concentration is higher than +0.5 standard deviation. Moreover, the PM$_{2.5}$ shows no association with birth weight when its concentration is lower than 
0 standard deviation, with sulfate concentration lower than -1 standard deviation. 
\begin{figure}[htp]
\centering
\includegraphics[scale=.62]{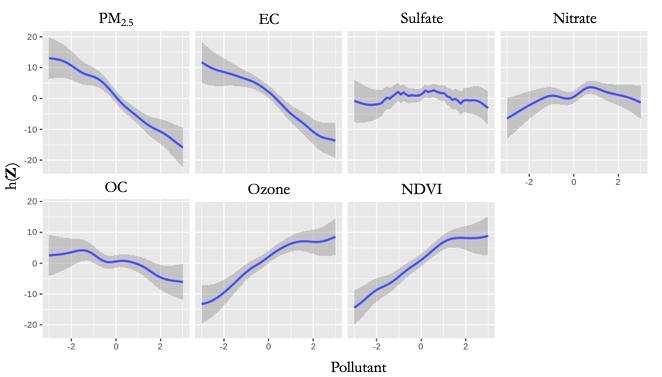}

\caption{Univariate estimated effects on birth-weight per standard deviation increase in PM$_{2.5}$, EC, OC, nitrate, sulfate, NDVI, and Ozone. 95\% confidence bands of estimates are in gray. All of the other mixture components are fixed to 50th percentile level when investigating single mixture effect on birth-weight. We show h(Z): difference of birth-weight comparing to the mean birth-weight of samples in grams; Pollutant(Z): change of each of the major constituents with the measure of standard deviation of that constituent.} 
\label{fig: pollutant exp}
\end{figure}

\begin{figure}[htp]
\centering
\includegraphics[scale=.62]{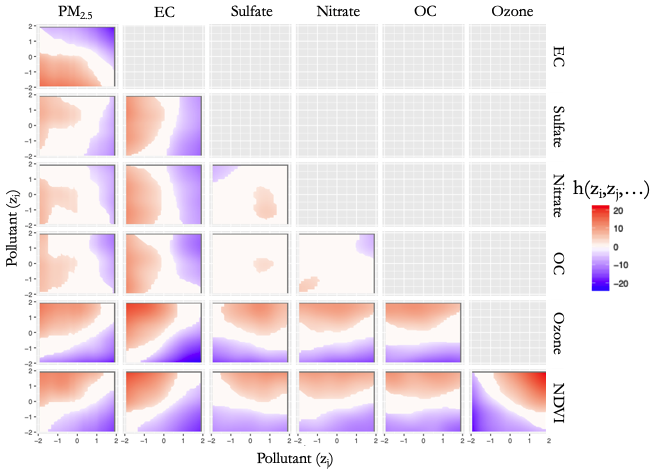}
\caption{Bivariate estimated effects on birthweight per standard deviation increase between PM$_{2.5}$, EC, OC, nitrate, sulfate, NDVI, and Ozone. All of the other mixture components are fixed to 50$^{th}$ percentile level when investigating bivariate mixture effects on birthweight.  We show h($z_i$, $z_j$): difference of birth-weight compared to the mean birth-weight of samples in grams; Pollutant($z_i$) and Pollutant($z_j$): change of each of the major constituents with the measure of standard deviation of that constituent.} 
\label{fig: pollutant bivariate}
\end{figure}

\section{Discussion}
As industry and governments invest in new technologies that ameliorate their environmental and pollution impacts, the need to quantify the effects of environmental mixtures on health is increasingly of interest. In parallel, electronic data registries such as the Massachusetts birth registry are increasingly being used in environmental health studies. These rich data sets allow measuring small, potentially non-linear effects of pollutant mixtures that impact public health. 
To the best of our knowledge, we propose the first semi-parametric Gaussian process regression framework that can be used to estimate effects using large datasets. In particular, we model the exposure-health outcome surface with a Gaussian process, and our applied model utilizes priors that allow for feature selection. Additionally, we use a linear component to incorporate confounder effects. Previous approaches for similar analysis had to either assume a parametric relationship or use a single pollutant per regression to estimate effects of interest \citep{FongKelvinC2019Rtom}.

To ameliorate the computational burden of computing the Bayesian posteriors of the Gaussian process, we propose a divide-and-conquer approach. Our method consists of splitting samples into subsets, computing the posterior distribution for each data split, and then combining the samples using a generalized median based on the order 1 Wasserstein distance \citep{Minsker2017}.    

We tailor the method to incorporate a squared exponential kernel and provide theoretical guarantees for the convergence of this foundational Gaussian Process component, which validates the soundness of our divide-and-conquer strategy for the non-parametric function. Our convergence results accommodate different assumptions for the underlying space of the true feature-response function. We provide theoretical and empirical results which illustrate a trade-off for the optimal number of splits. As the number of data splits increases, the posterior computation of the small data subsets will be faster; however, these posteriors will be noisy. In other words, there is a tension between computational cost and obtaining precise estimates. We propose using $K= n^{1/2}$ sample splits to efficiently approximate the posterior in a relatively short time. 

To illustrate the benefit of our method, we analyzed the association of a mixture of pollution constituents and green space and birthweights in the  Massachusetts birth registry. To our knowledge, a Gaussian process regression analysis, commonly known as Bayesian kernel machine regression in the environmental literature, of the Massachusetts birth registry data would not have been possible using existing fitting algorithms. Our analysis found the strongest adverse associations between traffic-related particles and PM$_{2.5}$ not accounted for by the other pollutants.  We observed the strongest positive associations with birthweight for Ozone and greenspace. A possible explanation for the Ozone finding is that Ozone is often typically negatively correlated with NO$_2$, another traffic pollutant, which was not included in this analysis. The greenspace finding is consistent with other analyses showing potentially beneficial effects of maternal exposure to greenness.

Our work also has some theoretical limitations. While our primary focus was on the computational scalability and the estimation of the non-linear function $h_0$, we did not include an analysis of the posterior distributions for the variable selection parameters ($r_j$). A full validation of the distributed model's variable selection properties is an important next step, but was beyond the scope of this paper, which is focused on solving the computational bottleneck of the GP component. In terms of applications, future efforts will also explore other multi-pollutant analyses both in this dataset as well as a range of other analyses in large-scale birth registry databases. 

\section*{Acknowledgments}
The first two authors contributed equally to this work. We thank Dr Luke Duttweiler for his work that improved the overall quality of the paper. We thank the referees and the associate editor for their constructive comments.

\section*{Funding}
Brent Coull was supported by the National Institutes of Health (NIH) under grant R01ES035735. Joel Schwartz was supported by the NIH under grant R01ES032418. Junwei Lu was supported by the National Science Foundation (NSF) under grant DMS-2434664, the William F. Milton Fund, and NIH R01ES032418.

\section*{Supplementary Materials}
Web Appendices, Tables, and Figures referenced in Sections 3, 4, and 5, are available with this paper at the Biometrics website on Oxford Academic. The simulation code is available online, and the corresponding package is available at https://github.com/junwei-lu/fbkmr.

\section*{Data Availability Statement}
Due to data privacy, we will only share the source code in the github repository in https://github.com/junwei-lu/fbkmr. Access to real data is available from the author, Joel Schwartz, upon reasonable request at jschwrtz@hsph.harvard.edu.

\newpage

\bibliographystyle{biom} 
\bibliography{paper-ref.bib}


\newpage
\setcounter{page}{1}

\begin{center}
\textit{\large Supplementary Materials for}
\end{center}
\begin{center}
{\large Scalable Gaussian Process Regression Via Median Posterior Inference for Estimating Multi-Pollutant Mixture Health Effects}
\vskip10pt
 Aaron Sonabend-W, Jiangshan Zhang, Luke Duttweiler, Edgar Castro, Joel Schwartz, Brent A. Coull, and Junwei Lu
\end{center}

\setcounter{section}{0}
\renewcommand{\thesection}{\Alph{section}}

\maketitle
\begin{abstract}
  This document contains the supplementary material to the paper
  ``Scalable Gaussian Process Regression Via Median Posterior Inference for Estimating Multi-Pollutant Mixture Health Effects".
\end{abstract}

\section{Proof of Theorem \ref{theorem: median rate}}\label{section: proof}
The proof for Theorem \ref{theorem: median rate} uses the following results.
\begin{assumption}\label{assumption: A general}
Let $A$ be a random variable with positive support, the distribution of $A$ has a Lebesgue density $g$ such that 
\[
C_1a^{d_1}\exp\{-D_1a^q\log^{d_2} a\}\le g(a)\le C_2 a^{d_1}\exp\{-D_2a^q\log^{d_2} a\},
\]
for every large enough $a>0$ and constants $C_1,D_1,C_2,D_2>0$ and $d_1,d_2\ge0$. 
\end{assumption}
\begin{assumption}\label{assumption: subexponential tails}
Let $H$ be a Gaussian field, the associated spectral measure $\mu$ satisfies
\[
\int e^{\delta\|\lambda\|}\mu(\delta\lambda)<\infty,
\]
for some $\delta>0$. We say that $H$ has subexponential tails.
\end{assumption}
\begin{assumption}\label{assumption: decreasing Lebesgue measure}
Let $H$ be a Gaussian field, $H$ possesses a Lebesgue density $f$ such that $a\mapsto f(a\lambda)$ is decreasing on $(0,\infty)$ for every, $\lambda\in\mathbb{R}^q$.
\end{assumption}

For a random variable $A$ satisfying Assumption \ref{assumption: A general}, let $H^A=\{H_{A\bz}:\bz\in[0,1]^q\}$ be a centered rescaled Gaussian process. We consider the Borel measurable map in $C[0,1]^q$, equipped with the uniform norm $\|\cdot\|_\infty$. 

\begin{theorem}[Theorem 3.1 in \citep{van_der_Vaart_2009}]\label{theorem: scaled SE conditions}
Let $H$ be a centered homogeneous Gaussian field which satisfies Assumptions \ref{assumption: subexponential tails}, \ref{assumption: decreasing Lebesgue measure}, then there exist Borel measurable subsets $B_n$ of $C[0,1]^q$ and for sufficiently large $n$ and big enough constant $C_4$
\begin{align}\label{cond_contract}
\begin{split}
\log N\left(\bar\epsilon_n,B_n,\|\cdot\|_\infty\right)\le  n\bar\epsilon_n^2,\\
P\left(H^A\notin B_n\right)\le  e^{-4n\epsilon_n^2},\\
P\left(\|H^A-h_0\|_\infty< \epsilon_n\right)\ge e^{-n\epsilon_n^2},
\end{split}
\end{align}
where
\begin{itemize}
\item if $h_0\in C^\alpha[0,1]^q$ for $\alpha>0$ then $\epsilon_n=n^{-\alpha/(2\alpha+q)}\left(\log n\right)^{\kappa_1}$, $\bar\epsilon=C_4\epsilon_n\left(\log n \right)^{\kappa_2}$, for $\kappa_1=((1+q\lor d_2)/(2+q/\alpha)$ and $\kappa_2=(1+q-d_2)/2$,
\item if $h_0$ is the restriction of a function in $\mathcal{A}^{\gamma,r}(\mathbb{R}^q)$ to $[0,1]^q$ with spectral density satisfying $|f(\lambda)|\ge C_3\exp\{-D_3\|\lambda\|^\nu\}$ for some constants $C_3,D_3,\nu>0$, then $\bar\epsilon_n=\epsilon_n(\log n)^{(q+1)/2}$ and $\epsilon_n=C_4n^{-\frac{1}{2}}\left(\log n\right)^{(q+1)/2}$ if $r\ge\nu$, and $\epsilon_n=C_4n^{-\frac{1}{2}}\left(\log n\right)^{\left((q+1)/2+q/(2r)\right)}$ if $r<2.$
\end{itemize}
\end{theorem}

\begin{theorem}[Theorem 2.2 in \citep{van_der_Vaart_2009}]\label{theorem: scaled gaussian rate}
Let $H=\{H_{\bZ}:\bz\in[0,1]^q\}$ be the centered Gaussian process, with covariance function $\mathbb{E}[H_{\bZ}H_{\bZ'}]=\exp\{-\|\bZ-\bZ'\|_n^2 \}$. Also let $\sigma$ be Gamma-distributed random variable. We consider $H=\{H_{\sigma\bZ}:\bz\in[0,1]^q\}$ as a prior distribution for $h_0$. Then for every large enough $M$,
\[
\Pi_{h,\sigma}\left(\|h-h_0\|_n+|\sigma-\sigma_0|\ge M\epsilon|\bZ_1,\dots,\bZ_n\right)\rightarrow0\text{ as }n\rightarrow\infty,
\]
where
\[
  \epsilon =
  \begin{cases}
    n^{-\left(\frac{\alpha}{2\alpha+q}\right)}\left(\log n\right)^{\left(\frac{4\alpha+2}{4\alpha+2q}\right)} & \text{if } h_0\in\mathcal{C}^\alpha[0,1]^q, \\
    n^{-\frac{1}{2}}\left(\log n\right)^{\left(q+1+q/(2r)\right)} & \text{if }h_0\in\mathcal{A}^{\gamma,r}(\mathbb{R}^q) \text{ and }r<2,\\
    n^{-\frac{1}{2}}\left(\log n\right)^{\left(q+1\right)} & \text{if }h_0\in\mathcal{A}^{\gamma,r}(\mathbb{R}^d) \text{ and }r\ge2.
  \end{cases}
\]
\end{theorem}

\begin{theorem}[Theorem 7 in \citep{Minsker2017}]\label{theorem: wasserstein bound}
Let $\bZ_1,\dots,\bZ_{n_k}\sim P_0$ be an $i.i.d$ sample, and assume that $\epsilon_k>0$ $\Theta_k\subset\Theta$ are such that for some constant $\tilde C_1>0$
\begin{align}\label{modified conditions}
\begin{split}
&\log M(\epsilon_k,\Theta_k,\psi)\le n_k\epsilon_k^2,\\
&\Pi(\Theta\backslash \Theta_k)\le\exp\{-n_k\epsilon_k^2(\tilde C_1+4)\},\\
&\Pi\left(\theta:-P_0\left(\log\frac{p_\theta}{p_0}\right)\le\epsilon_k^2,P_0\left(\log\frac{p_\theta}{p_0}\right)^2\le\epsilon_k^2\right)\ge\exp\{-\tilde C_1n_k\epsilon_k^2\}.
\end{split}
\end{align}
Then there exists constants $R(\tilde C_1)$ and $\tilde C_2$ such that 
\[
P\left(d_{W_{1,\rho}}\left(\delta_0,\Pi_k(\cdot\mid \bZ_1,\dots,\bZ_K)\right)\ge R\epsilon_k+e^{-\tilde C_2 n_k \epsilon_k^2}\right)\le\frac{1}{n_k\epsilon_k^2}+4e^{-\tilde C_2 n_k\epsilon_k^2}.
\]
\end{theorem}
\begin{corollary}[Corollary 8 in \citep{Minsker2017}]\label{corollary: median bound}
Let $\bZ_1,\dots,\bZ_n\sim P_0$ be an $i.i.d.$ sample, and let $\hat\Pi_{n,g}$ be defined as in Theorem \ref{theorem: median rate}.  Under conditions \ref{modified conditions}, if $\epsilon_k$ satisfies $1/(n_k\epsilon_k^2)+4e^{-(1+\tilde C_2/2) n_k\epsilon_k^2/2}<\frac{1}{7}$, then
\[
\mathbb{P}\left(\left\|\delta_0-\hat\Pi_{n,g}\right\|_\mathcal{F}\ge1.52\left(R\epsilon_k+e^{-\tilde C_2 n_k \epsilon_k^2}\right)\right)<1.27^{-K}
\]
\end{corollary}
\begin{proof}(of Theorem \ref{theorem: median rate}):
If $\frac{1}{\rho^q} = A^q$ has a Gamma distribution, then Assumption \ref{assumption: A general} is satisfied for our model with $d_2=0$. Additionally, as $H$ is squared exponential Gaussian process, it is a density relative to the Lebesgue measure given by 
\[
\lambda\mapsto\frac{1}{2^q\pi^{q/2}}\exp\{-\|\lambda\|^2/4\}
\]
which has sub-exponential tails (see \citep{van_der_Vaart_2009}). Therefore, by Theorem \ref{theorem: scaled SE conditions} conditions \eqref{cond_contract} are satisfied for $H^A$ with $\epsilon_k =n_k^{-\alpha/(2\alpha+q)}\left(\log n_k\right)^{(4\alpha+q)/(4\alpha+2q)}$ if $h_0\in\mathcal{C}[0,1]^q$ and if $h_0\in\mathcal{A}^{\gamma,r}(\mathbb{R}^q)$, then 

\[
 \epsilon_k= 
  \begin{cases}
    n_k^{-\frac{1}{2}}\left(\log n_k\right)^{\left(q+1+q/(2r)\right)} & \text{if }r<2,\\
    n_k^{-\frac{1}{2}}\left(\log n_k\right)^{\left(q+1\right)} & \text{if }r\ge2.
  \end{cases}
\]
Note that \eqref{cond_contract} (from non-parametric GP contraction theory) map one-to-one to Conditions in \eqref{modified conditions} for the Hellinger metric $\psi$ (from the general median posterior theory ) (see \citep{VaartvanderA2008}), thus by Theorem \ref{theorem: wasserstein bound} with $\epsilon_k>0$ defined as above we have
\begin{align}\label{equation: prob. bound}
P\left(d_{W_{1,\rho}}\left(\delta_0,\Pi_k(\cdot\mid \bZ_1,\dots,\bZ_K)\right)\ge R\epsilon_k+e^{-\tilde C_2 k\epsilon_k^2}\right)\le\frac{1}{n_k\epsilon_k^2}+4e^{-\tilde C_2 n_k\epsilon_k^2}.
\end{align}
note that whether $h_0\in\mathcal{C}^\alpha[0,1]^q$ or $h_0\in\mathcal{A}^{\gamma,r}(\mathbb{R}^q)$ we can choose $k(n)$ such that $1/(n_k\epsilon_k^2)+4e^{-(1+\tilde C_2/2) n_k\epsilon_k^2/2}<\frac{1}{7}$. For example any $k\le n^{1/2}\log n$ would work well. Therefore, for any $\delta>0$, and fixed $k(n)$ using Corollary \ref{corollary: median bound} there is an $\epsilon_k(\delta)$ with a large enough $n$ such that
\[
\mathbb{P}\left(\left\|\delta_0-\hat\Pi_{n,g}\right\|_\mathcal{F}\ge1.52\left(R\epsilon_k+e^{-\tilde C_2 n_k \epsilon_k^2}\right)\right)<\delta.
\]
\end{proof}
\section{Details on Application to Boston Birth Weight Data}\label{appendix: birthweight data}
Each record consists of the outcome of interest which is the birth-weight of the newborn, confounders such as maternal age (years), maternal race (white, black, Asian, American Indian, other), maternal marital status (married, not married), maternal smoking during or before pregnancy (yes, no), maternal education (highest level of education attained: less than high school, high school, some college, college, advanced degree beyond college), parity (first-born, not first-born), maternal diabetes (yes, no), gestational diabetes (yes, no), maternal chronic high blood pressure (yes, no), maternal high blood pressure during pregnancy (yes, no), Kessner index of adequacy of prenatal care (adequate, intermediate, inadequate, no prenatal care), mode of delivery (vaginal, forceps, vacuum, first cesarean birth, repeat cesarean birth, vaginal birth after cesarean birth), clinical gestational age (weeks), year of birth (one of 2001–2012), season of birth (spring, summer, autumn, winter), date of birth,  newborn sex (male, female), Ozone concentration, Normalized Difference Vegetation Index (NDVI),  Medicaid-supported prenatal care (yes, no). Finally pollution exposure measures are concentration of PM$_{2.5}$ and four major chemical constituents of it: elemental carbon (EC), organic carbon (OC), nitrate, and sulfate. After excluding the observation records with missing data, the final sample with size equal to $n=685,857$ is used for our model illustration. We treated normalized Ozone , NDVI, PM$_{2.5}$, EC, OC, nitrate, and sulfate as mixture components for non-parametric parts, and other variables as covariates. For the date of birth within one year, in order to control the temporal effect on birth weight, we implement a cosine transformation on it, with birth date on January $1^{st}$ has highest positive effect on birth weight, and June $15^{th}$ has lowest negative effect on the birth weight. The model used is \eqref{eq: model}. In our main analysis, we scaled the estimated effects per a standard deviation increase per each pollutant, which is more representative of a real world scenario than mass scaling.
\begin{figure}
    \centering
    \includegraphics[width=0.5\linewidth]{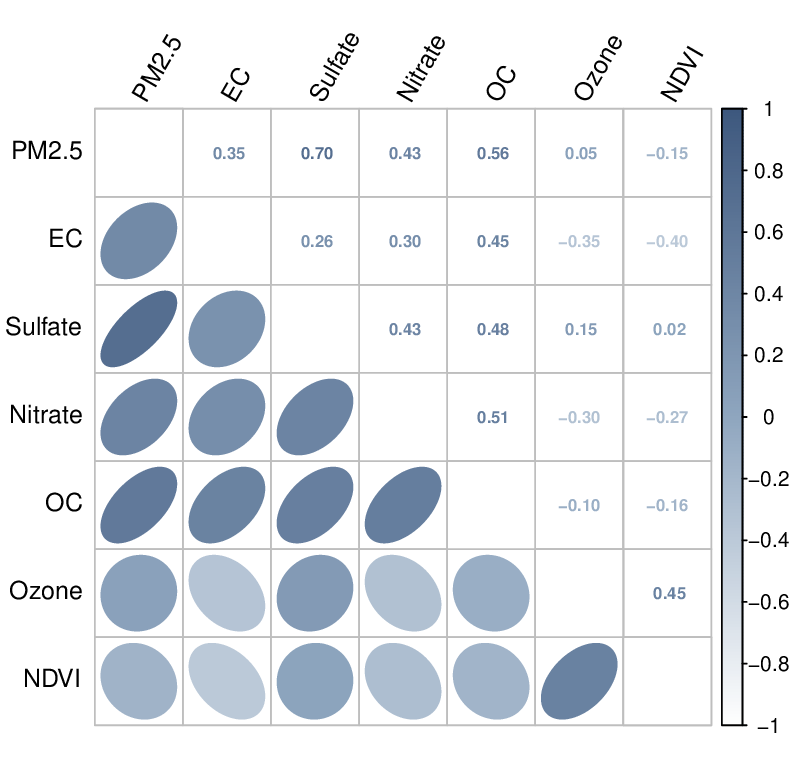}
    \caption{Correlation matrix of Boston birth weight data for 7 mixture component.}
    \label{fig:corr_pollution}
\end{figure}
\section{Additional Simulation Result}\label{appendix: additional simulation results}
\subsection{Performance of method with correlated exposures}    
In order to extend the method to a more realistic exposure setting where exposures are highly correlated, we consider that $z_1$ and $z_3$ are highly correlated with a correlation $\rho_{13}=0.8$; in the meantime, correlation between $z_2$ and $z_3$ is set as $\rho_{23}=0.3$; correlation between $z_1$ and $z_2$ is set as $\rho_{12}=0.1$. We applied the same analysis pipeline as in the independent-exposure setting in section \ref{simu split levels}, and report results in Figure \ref{fig:h_corr} and \ref{fig:h_mse_corr}. It can be seen that compared to the result without exposure correlation, the inference performance of the method on $h$ is similar. However, the empirical standard errors for all metrics are larger than those in the independent-exposure setting, as expected in the presence of multicollinearity.We also observe that component-wise variable selection can become unstable when exposures are highly correlated, leading to increased variability in PIP estimates. As a potential extension, we note that hierarchical/group-level selection could mitigate this issue within our framework. Finally, the computation–precision trade-off remains evident, with the optimal choice of $t \xrightarrow{}1/2$ as n increases.
    \begin{figure}[h]
   \centering
   \includegraphics[width=0.8\linewidth]{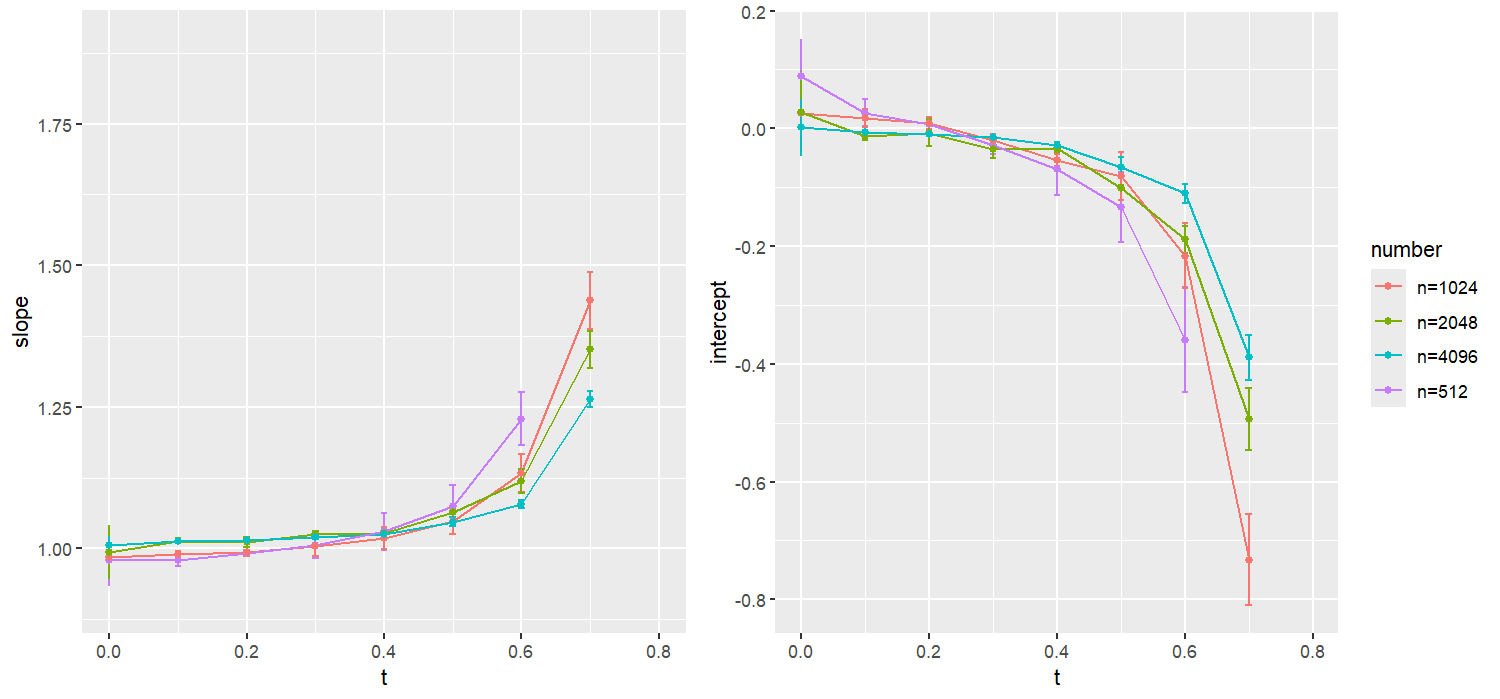}
    \caption{Regression summary results for ${\bf h} = \gamma_0 + \gamma_1 {\bf\hat h}$ across different sample size $n$ and data set splits under correlated exposure setting. The setting of number of subsets are described above as $n^t$. We show (A) intercept: $\hat\gamma_0$, (B) slope: $\hat\gamma_1$. } 
    \label{fig:h_corr}
\end{figure}

\begin{figure}[h]
    \centering
    \includegraphics[width=0.8\linewidth]{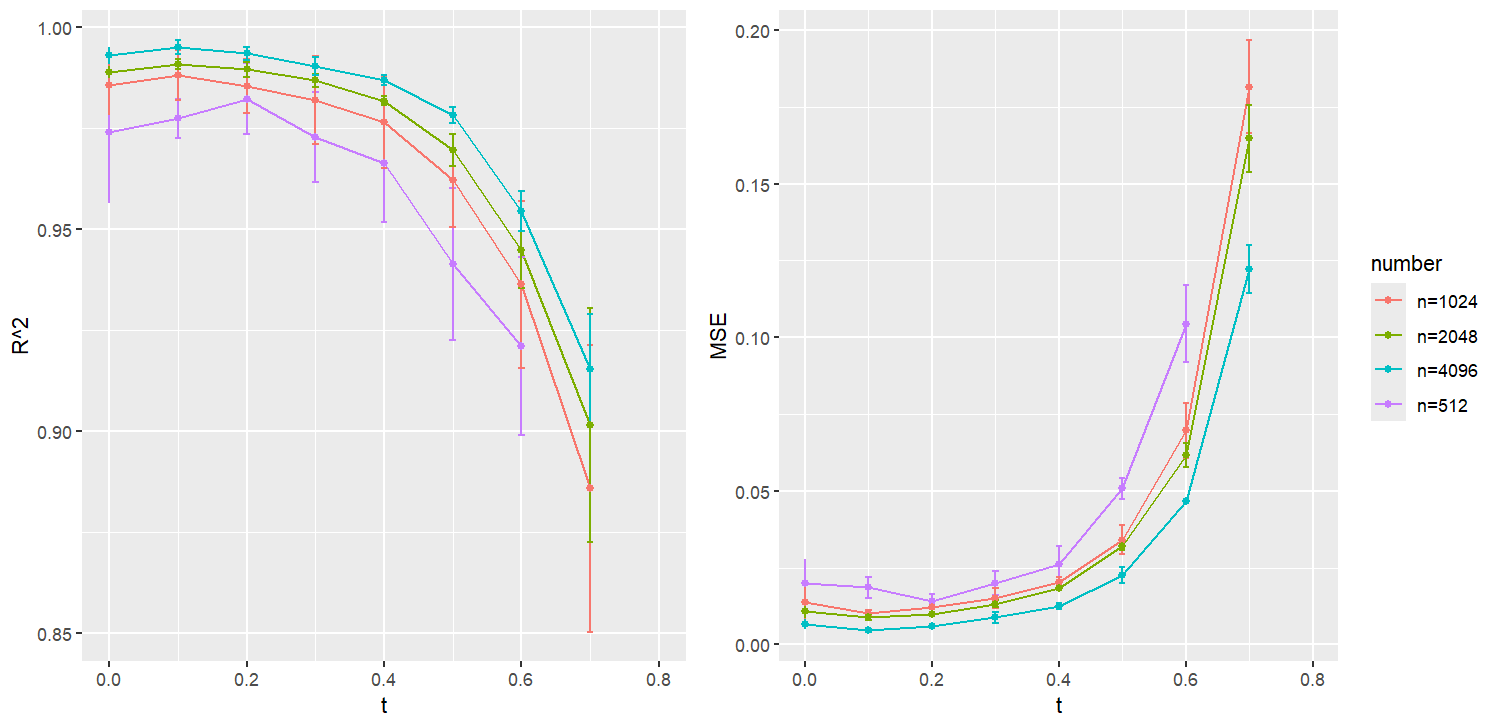}
    \caption{(A)Regression $R^2$ for ${\bf h} = \gamma_0 + \gamma_1 {\bf\hat h}$ and (B)MSE of ${\bf\hat h}$ for the proposed method across different sample size $n$ and data set splits under correlated exposure setting. The setting of number of subsets are described above as $n^t$. }
    \label{fig:h_mse_corr}
\end{figure}

\subsection{Performance of method with extreme uneven partition}
Now we would like to discuss a more extreme setting of an uneven partitioning. Following the simulation setting in section \ref{simu subsetting scheme}, we now consider that the sample size for each subset is randomly sampled from $[n^{(1-t-\Delta)}, \min(n^{(1-t+\Delta)},n)]$, with $\Delta \in \{0,0.1,0.2,0.3\}$. In this setting, the size of the subset varies in an exponential order. As shown in Table \ref{tab:mse_delta}, the MSE of $\hat{h}$ increases when $\Delta$ increases, while it remains at a consistent range when the minimal possible sample size of the subset $n^{(1-t-\Delta)}$ is greater than $n^{1/2}$. Once sample sizes of some of the subsets are less than $n^{1/2}$, the inference performance of the method drops dramatically. This simulation result is consistent with our theoretical finding. 
\begin{table}[!htp]
    \centering
    \begin{tabular}{ccccc}
        t & \multicolumn{4}{c}{$\Delta$}  \\
         ~& 0 &0.1 &0.2 &0.3\\
         \hline
         0.1 & 0.0064(0.0007) & 0.0067(0.0007) & 0.0092(0.0028) & 0.0124(0.0047)\\
         0.2 & 0.0077(0.0008) & 0.0085(0.0015) & 0.0120(0.0031) & 0.0162(0.0066) \\
         0.3 & 0.0104(0.0009) & 0.0125(0.0017) & 0.0219(0.0055) & 0.0355(0.0101) \\
         0.4 & 0.0155(0.0012) &0.0204(0.0020) & 0.0418(0.0083) & 0.1426(0.0477) 
    \end{tabular}
    \caption{MSE of $\hat{h}$ across different data set splits and sample size variance parameter $\Delta$.}
    \label{tab:mse_delta}
\end{table}

\begin{table}
    \centering
    \begin{tabular}{ccccccccc}
      \hline
      \hline
         \multicolumn{9}{c}{$X_3$}\\
         \hline
         \hline
        n & \multicolumn{8}{c}{t}  \\
         ~& 0 &0.1 &0.2 &0.3 &0.4 &0.5 &0.6 &0.7\\
         \hline
         512 & 0.018 & 0.016 & 0.024 & 0.132 &  0.131 & 0.121 & 0.198 & NA\\
         1024 & 0.015 & 0.015 &0.018 & 0.0034 & 0.067 & 0.076 & 0.098 & 0.054\\
         2048 & 0.004 & 0.004 & 0.006 & 0.008 & 0.061 & 0.045 & 0.052 & 0.045\\
         4096  & 0.001 & 0.002 & 0.003 & 0.006 & 0.008 & 0.023 & 0.038 &0.045\\
     \hline
     \hline
         \multicolumn{9}{c}{$X_4$}\\
         \hline
         \hline
         n & \multicolumn{8}{c}{t}  \\
         ~& 0 &0.1 &0.2 &0.3 &0.4 &0.5 &0.6 &0.7\\
         \hline
         512 & 0.070 & 0.074 & 0.095 & 0.155 & 0.177 & 0.199 & 0.191 & NA\\
         1024 & 0.004 & 0.006 & 0.003 &0.008 & 0.076 & 0.059 & 0.070  &0.068\\
         2048 & 0.004 & 0.003 & 0.004 & 0.009 & 0.057 & 0.066 & 0.057 & 0.054\\
         4096 & 0.002 &0.002 & 0.003 & 0.004 & 0.024 & 0.024 & 0.061 & 0.063
    \end{tabular}
    \caption{Mean estimated PIP of $X_3$ and $X_4$}
    \label{tab:PIP}
\end{table}

\begin{figure}
        \centering
        \includegraphics[width=0.5\linewidth]{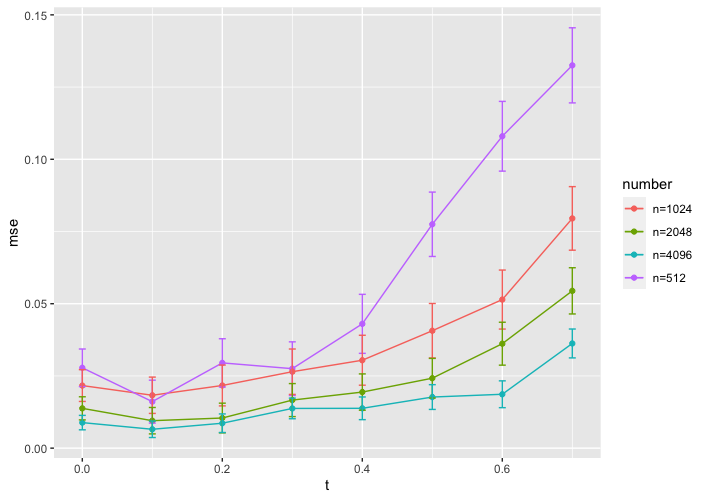}
        \caption{MSE of $\hat{h}$ across different sample size n and data set splits.}
        \label{fig:mse}
\end{figure}
    
\begin{figure}
    \centering
    \includegraphics[width=0.5\linewidth]{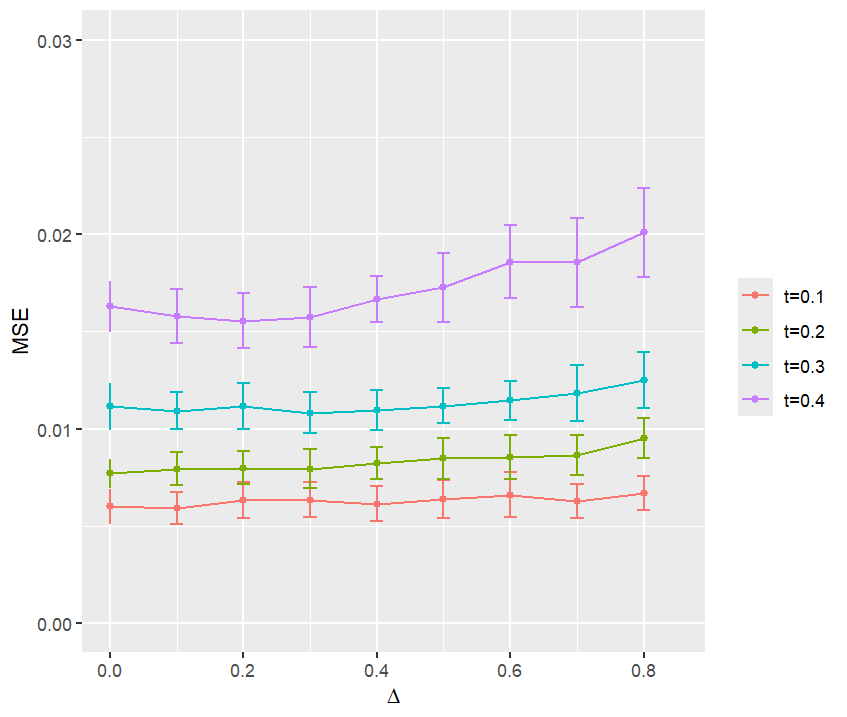}
    \caption{MSE of $\hat{h}$ across different data set splits and sample size variance parameter $\Delta$.}
    \label{fig:mse_split_h}
\end{figure}


\end{document}